\documentclass[english]{article}
\usepackage{amsfonts,amsmath,amsthm,amscd,amssymb,latexsym}
\usepackage[active]{srcltx}                    

\theoremstyle{plain}
\newtheorem{thm}{Theorem}[section]

\newtheorem{prop}[thm]{Proposition}

\theoremstyle{definition}
\theoremstyle{remark}
\numberwithin{equation}{section}

\newcommand{\subjclass}{\textbf{Math. Subj. Clas.: }\medskip}

\title{On the chirality of a discrete Dirac-K\"{a}hler equation  }
\author{ Volodymyr Sushch \\ Koszalin University of Technology, \\ Sniadeckich 2, 75-453 Koszalin, Poland \\ e-mail: volodymyr.sushch@tu.koszalin.pl  }

\begin{document}

\maketitle
\begin{abstract}
We discuss a discrete analogue of the Dirac-K\"{a}hler equation in which chiral properties of the continual counterpart are captured.
We pay special attention to a discrete Hodge star operator. To build one a combinatorial construction of double complex is used.
We describe discrete exterior calculus operations on a double comlex and obtain the discrete Dirac-K\"{a}hler equation using these tools.
 Self-dual and anti-self-dual discrete inhomogeneous forms  are presented.  The chiral invariance of the massless  discrete Dirac-K\"{a}hler  equation is shown. Moreover,
  in the massive case we prove that a discrete  Dirac-K\"{a}hler operator flips the chirality.
\end{abstract}



\section{Introduction}
This work is a direct continuation of the paper \cite{S3}  in which we constructed a new discrete analogue of the  Dirac-K\"{a}hler  equation.
In \cite{S3} we proposed a geometric discretisation scheme based on the formalism  of  differential forms and showed that many of the algebraic relations amongst the Hodge star operator  $\ast$,  the exterior product $\wedge$, the differential $d$,  and the adjoint $\delta$ of $d$ that hold in the smooth setting  also hold in the discrete case. This approach was originated by Dezin \cite{Dezin}. There are alternative geometric discretisation schemes which base on the use of the differential form language    \cite{Beauce1, Becher, Catterall2, Dodziuk, Kanamori, Rabin, SSSA, S1}. Difficulties of the discretisation of the Hodge operators have been described by several authors  \cite{Becher,  Rabin, SSSA, Teixeira1,  Teixeira, Wilson}.

In this paper we are going to discuss a decomposition of the discrete  Dirac-K\"{a}hler  equation into its self-dual and anti-self-dual parts. The reason that we did not
consider this problem in \cite{S3} is that the Hodge star operator $\ast$ and its  discrete analogue are slightly  differ. In the continual case the operator $(\ast)^2$
is either an involution or antiinvolution while in the discrete model $(\ast)^2$ is equivalent to a shift with corresponding sign. This is one of the main distinctive features of the formalism \cite{S3} as compared to the continual case.
Now we define a discrete star operator using a combinatorial double complex construction. In this way we obtain the operator which  is more like its continual analogue since  $(\ast)^2=\pm I$, where $I$ is the identity operator. At the same time in the double complex we will use discrete analogues of differential forms, $d$, and $\wedge$ defined in \cite{S3}. The discrete star operator proposed here still preserves the Lorentz metric structure in our discrete model. It make possible to define a discrete analogue of  $\delta$ by using  an inner product of discrete forms (cochains) which imitate the continual case.

From the physics point of view self-dual and anti-self-dual fields of Dirac theory  correlate  with chiral fermions. It is well known that the chirality is an important feature of the Dirac theory. However, in lattice formulations of  fermions the chirality problem is one of the most notorious. This problem deals with breaking chiral symmetry on the lattice  \cite{Catterall,Kaplan, Luscher, Narayanan, Susskind}. See also more resent work \cite{Kaplan1}. For the Dirac-K\"{a}hler  equation on the lattice this difficulty was discussed first by Rabin \cite{Rabin}. In \cite{Beauce} a discretisation scheme in the Dirac-K\"{a}hler setting was proposed in which the chirality is captured on the lattice.
Chiral fermions can be described by using the fifth gamma matrix: $\gamma^5=\mathrm{i}\gamma^0\gamma^1\gamma^2\gamma^3$. The projection operators
 $P_L=\frac{I-\gamma^5}{2}$  and $P_R=\frac{I+\gamma^5}{2}$ decompose any Dirac field into its left-handed and right-handed parts.
 Rabin \cite{Rabin} pointed out that in the language of  differential forms
the chiral symmetry is a rotation of mixing forms with their duals  and the Hodge star operator plays a central role here. However, the operator $\ast$ is somewhat different from $\gamma^5$. Therefore,  we introduce a modified star operator which plays a role of $\gamma^5$ in our discrete model. Self-dual and anti-self-dual discrete inhomogeneous forms with respect to this new star operator are considered.  From the viewpoint presented here this allows us to deal with chirality. We show that, just as in the continuum case, a discrete massless  Dirac-K\"{a}hler  equation admits the chiral invariance. We prove also that in the massive case a discrete analog of the Dirac-K\"{a}hler operator flips the chirality.

We first briefly review some definitions and basic notations on the Dirac-K\"{a}hler equation \cite{Kahler, Rabin}.
Let $M={\mathbb R}^{1,3}$ be  Minkowski space with  metric signature  $(+,-,-,-)$.
Denote by $\Lambda^r(M)$ the vector space of smooth differential $r$-form, $r=0,1,2,3,4$. We consider  $\Lambda^r(M)$ over $\mathbb{C}$.
 Let $\omega$ and  $\varphi$ be  complex-valued  $r$-forms on $M$. The inner product is defined by
\begin{equation}\label{1}
(\omega, \ \varphi)=\int_{M}\omega\wedge\ast\overline{\varphi},
\end{equation}
where $\wedge$ is the exterior product and $\ast$ is the Hodge star operator  $\ast:\Lambda^r(M)\rightarrow\Lambda^{4-r}(M)$ (with respect to the Lorentz metric).
Let $d:\Lambda^r(M)\rightarrow\Lambda^{r+1}(M)$ be the exterior differential and let $\delta:\Lambda^r(M)\rightarrow\Lambda^{r-1}(M)$ be the formal adjoint of $d$  with respect to (\ref{1}) (codifferential). We have $\delta=\ast d\ast$.  Then the Laplacian (Laplace-Beltrami operator) acting on $r$-forms is defined by
\begin{equation}\label{2}
\Delta\equiv -d\delta-\delta d:\Lambda^r(M)\rightarrow\Lambda^{r}(M).
\end{equation}
It is clear that
\begin{equation*}
-d\delta-\delta d=(d-\delta)^2=-(d+\delta)^2.
\end{equation*}
Hence, a square root of the Laplacian can be written in two ways, namely  either  $d-\delta$ or $\mathrm{i}(d+\delta)$, where $\mathrm{i}$ is the usual complex unit ($\mathrm{i}^2=-1$).
In this paper we formulate the complex Dirac-K\"{a}hler equation by using the operator $\mathrm{i}(d+\delta)$.
Denote by $\Lambda(M)$ the set of all differential forms on $M$. We have
$\Lambda(M)=\Lambda^0(M)\oplus\Lambda^1(M)\oplus\Lambda^2(M)\oplus\Lambda^3(M)\oplus\Lambda^4(M)$. Let $\Omega\in\Lambda(M)$
be an inhomogeneous differential form.   This form can be expanded  as
\begin{equation*}
\Omega=\sum_{r=0}^4\overset{r}{\omega},
\end{equation*}
where $\overset{r}{\omega}\in\Lambda^r(M)$.
 The Dirac-K\"{a}hler equation is given by
\begin{equation}\label{3}
\mathrm{i}(d+\delta)\Omega=m\Omega,
\end{equation}
where $m$  is a mass parameter.
It is easy to show that Eq. (\ref{3}) is equivalent to the set of equations
\begin{align*}\label{}
\mathrm{i}\delta\overset{1}{\omega}=m\overset{0}{\omega},\\
\mathrm{i}(d\overset{0}{\omega}+\delta\overset{2}{\omega})=m\overset{1}{\omega},\\
\mathrm{i}(d\overset{1}{\omega}+\delta\overset{3}{\omega})=m\overset{2}{\omega},\\
\mathrm{i}(d\overset{2}{\omega}+\delta\overset{4}{\omega})=m\overset{3}{\omega},\\
\mathrm{i}d\overset{3}{\omega}=m\overset{4}{\omega}.
\end{align*}

\section{Double complex construction}
In this section we introduce a double complex construction which bases on
the combinatorial model of Minkowski space described in \cite{S3}.
For the convenience of the reader we briefly repeat the relevant material from \cite{S3}
without proofs, thus making our presentation self-contained.

Let the tensor product  $C(4)=C\otimes C\otimes C\otimes C$
of a   1-dimensional complex be a combinatorial model of Euclidean space
 ${\mathbb R}^4$. The 1-dimensional complex $C$ is defined in the following way.
Let $C^0$ denotes the real linear space of 0-dimensional chains generated by
basis elements $x_\kappa$ (points), $\kappa\in {\mathbb Z}$. It is convenient to
introduce the shift operator  $\tau$ in the set of indices by
\begin{equation}\label{4}
\tau\kappa=\kappa+1.
\end{equation}
We denote the open interval $(x_\kappa, x_{\tau\kappa})$ by $e_\kappa$.
One can regard the set $\{e_{\kappa}\}$ as a set of basis elements of
the real linear space $C^1$. Suppose that $C^1$ is the space of 1-dimensional
chains.
Then the 1-dimensional complex (combinatorial real line) is the direct sum
of the introduced spaces $C=C^0\oplus C^1$. The boundary operator $\partial$
in $C$ is given by
$$
\partial x_\kappa=0, \qquad  \partial e_\kappa=x_{\tau\kappa}-x_\kappa.
$$
The definition is extended to arbitrary chains by linearity.

Multiplying the basis elements $x_\kappa, e_\kappa$ in various way we obtain
basis elements of $C(4)$. Let $s_k$ be  an arbitrary basis element of $C(4)$. Then we have
${s_k=s_{k_0}\otimes s_{k_1}\otimes s_{k_2}\otimes s_{k_3}}$, where $s_{k_i}$
is either  $x_{k_i}$ or  $e_{k_i}$ and $k_i\in {\mathbb Z}$. Here $k=(k_0, k_1, k_2, k_3)$ is a multi-index.
Let
\begin{equation*}
x_k=x_{k_0}\otimes x_{k_1}\otimes x_{k_2}\otimes x_{k_3}  \quad \mbox{and} \quad
e_k=e_{k_0}\otimes e_{k_1}\otimes e_{k_2}\otimes e_{k_3}
\end{equation*}
be the 0- and 4-dimensional basis elements
of $C(4)$.
The 1-dimensional basis elements
of $C(4)$ can be written as
\begin{align}\label{5}
e_k^0=e_{k_0}\otimes x_{k_1}\otimes x_{k_2}\otimes x_{k_3}, \nonumber \qquad
e_k^1=x_{k_0}\otimes e_{k_1}\otimes x_{k_2}\otimes x_{k_3}, \\
e_k^2=x_{k_0}\otimes x_{k_1}\otimes e_{k_2}\otimes x_{k_3},  \qquad
e_k^3=x_{k_0}\otimes x_{k_1}\otimes x_{k_2}\otimes e_{k_3},
\end{align}
where  the superscript $i$ indicates  a place of $e_{k_i}$ in $e_k^i$ and  $i=0,1,2,3$.
In the same way we will write  the 2-dimensional basis elements
of $C(4)$ as
\begin{align}\label{6}
e_k^{01}=e_{k_0}\otimes e_{k_1}\otimes x_{k_2}\otimes x_{k_3},  \qquad
e_k^{12}=x_{k_0}\otimes e_{k_1}\otimes e_{k_2}\otimes x_{k_3},\nonumber \\
e_k^{02}=e_{k_0}\otimes x_{k_1}\otimes e_{k_2}\otimes x_{k_3}, \qquad
e_k^{13}=x_{k_0}\otimes e_{k_1}\otimes x_{k_2}\otimes e_{k_3}, \nonumber \\
e_k^{03}=e_{k_0}\otimes x_{k_1}\otimes x_{k_2}\otimes e_{k_3}, \qquad
e_k^{23}=x_{k_0}\otimes x_{k_1}\otimes e_{k_2}\otimes e_{k_3}.
\end{align}
Finally, denote by $e_k^{012}, e_k^{013}, e_k^{023}, e_k^{123}$ the 3-dimensional basis elements
of $C(4)$.

Let $C(4)=C(p)\otimes C(q)$, where $p+q=4$. If $a\in C(p)$ and $b\in C(q)$ are arbitrary chains, belonging to the complexes being multiplied, then  we extend the definition of   $\partial$ to chains of $C(4)$ by the rule
\begin{equation}\label{7}
\partial(a\otimes b)=\partial a\otimes b+(-1)^ra\otimes\partial b,
\end{equation}
where $r$ is the dimension of the chain $a$. It is easy to check that  $\partial\partial a=0$.

Let us introduce the construction of a double complex. This construction generalizes that of \cite{S2}. Together with
the complex $C(4)$ we consider its double, namely the complex
$\tilde{C}(4)$ of exactly the same structure. Define the one-to-one
correspondence
\begin{equation*}
\ast^c : C(4)\rightarrow\tilde{C}(4), \qquad \ast^c : \tilde
C(4)\rightarrow C(4)
\end{equation*}
in the following way. Let  $s_k^{(r)}$ be an arbitrary
$r$-dimensional basis element of $C(4)$, i.e.,  the product
$s_k^{(r)}=s_{k_0}\otimes s_{k_1}\otimes s_{k_2}\otimes s_{k_3}$
contains exactly $r$ of $1$-dimensional elements $e_{k_i}$ and $4-r$
of $0$-dimensional elements  $x_{k_i}$,   and  the superscript $(r)$ indicates  a position  of $e_{k_i}$ in $s_k^{(r)}$ (see, for example,  (\ref{5}) and (\ref{6})).  Then
\begin{equation}\label{8}
\ast^cs_k^{(r)}=\varepsilon(r)\tilde s_k^{(4-r)}, \qquad \ast^c\tilde s_k^{(r)}=\varepsilon(r) s_k^{(4-r)},
\end{equation}
where
\begin{equation*}
 \tilde s_k^{(4-r)}=\ast^c s_{k_0}\otimes \ast^c s_{k_1}\otimes
\ast^c s_{k_2}\otimes \ast^c s_{k_3}
\end{equation*}
and $\ast^c s_{k_i}=\tilde e_{k_i}$ if $s_{k_i}=x_{k_i}$ and
$\ast^c s_{k_i}=\tilde x_{k_i}$ if $s_{k_i}=e_{k_i}.$
Here $\varepsilon(r)$ is the Levi-Civita symbol, i.e.,
\begin{equation*}
\varepsilon(r)=\left\{\begin{array}{l}+1 \quad  \mbox{if} \quad ((r),
(4-r))\quad \mbox{is an even permutation of} \quad (0,1,2,3) \\
                            -1 \quad  \mbox{if} \quad ((r),
(4-r))\quad \mbox{is an odd permutation of} \quad (0,1,2,3).
                            \end{array}\right.
\end{equation*}
 For example, for the 1- and 2-dimensional basis elements  we have
\begin{equation*}\label{}
\ast^c e_k^0=\tilde e_k^{123},  \qquad \ast^c e_k^1=-\tilde e_k^{023}, \qquad
\ast^c e_k^2=\tilde e_k^{013}, \qquad \ast^c e_k^3=-\tilde e_k^{012}
\end{equation*}
and
\begin{align*}
\ast^c e_k^{01}=\tilde e_k^{23},  \qquad \ast^c e_k^{02}=-\tilde e_k^{13}, \qquad  \ast^c e_k^{03}=\tilde e_k^{12}, \\
\ast^c e_k^{12}=\tilde e_k^{03}, \qquad \ast^c e_k^{13}=-\tilde e_k^{02}, \qquad \ast^c e_k^{23}=\tilde e_k^{01}.
\end{align*}
We will also use the symbol $\varepsilon'(r)$ to denote the Levi-Civita symbol given by
\begin{equation*}
\varepsilon'(r)=\left\{\begin{array}{l}+1 \quad  \mbox{if} \quad ((4-r), (r))\quad \mbox{is an even permutation of} \quad (0,1,2,3) \\
                            -1 \quad  \mbox{if} \quad ((4-r), (r))\quad \mbox{is an odd permutation of} \quad (0,1,2,3).
                            \end{array}\right.
\end{equation*}
So we can write
\begin{equation*}
\ast^c s_k^{(4-r)}=\varepsilon'(r)\tilde s_k^{(r)} \qquad \mbox{and} \qquad \ast^c\tilde s_k^{(4-r)}=\varepsilon'(r)s_k^{(r)}.
\end{equation*}
It is easy to check that
\begin{equation}\label{9}
\varepsilon(r)\varepsilon'(r)=(-1)^r.
\end{equation}

\begin{prop}Let $a_r\in C(4)$ be an $r$-dimensional chain, i.e.,
\begin{equation}\label{10}
a_r=\sum_k\sum_ra_{(r)}^ks_k^{(r)}, \qquad  a_{(r)}^k\in\mathbb R.
\end{equation}
Then we have
\begin{equation}\label{11}
\ast^c\ast^c
 a_r=(-1)^{r}a_r.
\end{equation}
\end{prop}
 \begin{proof} The proof consists in applying the operation $\ast^c$ for basis elements and using (\ref{9}).
\end{proof}
Suppose that the combinatorial model of Minkowski space has the same structure
as $C(4)$.  We will use the index $k_0$ to denote the basis elements of $C$ which correspond to the time coordinate of
$M$. Hence, the indicated basis elements will be written as  $x_{k_0}$, $e_{k_0}$.

Let us now consider a dual complex to $C(4)$. We define its as the complex of cochains
$K(4)$ with complex coefficients. The complex $K(4)$
has a similar structure, namely ${K(4)=K\otimes K\otimes K\otimes K}$, where $K$ is a dual
complex to the 1-dimensional complex $C$. We will write the basis elements of $K$  as
$x^\kappa$ and $e^\kappa$, $\kappa\in {\mathbb Z}$. Then an arbitrary $r$-dimensional basis element of $K(4)$ can be written  as
$s_{(r)}^k=s^{k_0}\otimes s^{k_1}\otimes s^{k_2}\otimes s^{k_3}$, where $s^{k_i}$
is either  $x^{k_i}$ or  $e^{k_i}$. Note that  the superscript and the subscript change places,   as compared $s_{(r)}^k$ with the corresponding  basis element  $s^{(r)}_k$ of $C(4)$.  We will call cochains forms,
emphasizing their relationship with the corresponding continual objects, differential forms.

Denote by $K^r(4)$ the set of forms of degree $r$. We can represent  $K(4)$ as
\begin{equation*}
K(4)=K^0(4)\oplus K^1(4)\oplus K^2(4)\oplus K^3(4)\oplus K^4(4).
\end{equation*}
Let $\overset{r}{\omega}\in K^r(4)$. Then we have
\begin{equation}\label{12}
\overset{0}{\omega}=\sum_k\overset{0}{\omega}_kx^k, \qquad  \overset{4}{\omega}=\sum_k\overset{4}{\omega}_ke^k,
\end{equation}
\begin{equation}\label{13}
\overset{1}{\omega}=\sum_k\sum_{i=0}^3\omega_k^ie_i^k, \qquad
\overset{2}{\omega}=\sum_k\sum_{i<j} \omega_k^{ij}e_{ij}^k, \qquad
\overset{3}{\omega}=\sum_k\sum_{i<j<l} \omega_k^{ijl}e_{ijl}^k,
\end{equation}
where the components $\overset{0}{\omega}_k, \ \overset{4}{\omega}_k, \  \omega_k^i, \ \omega_k^{ij}$ and $\omega_k^{ijl}$ are complex numbers.

As in \cite{Dezin} and \cite{S3}, we define the pairing (chain-cochain) operation for any basis elements
$\varepsilon_k\in C(4)$,  $s^k\in K(4)$ by the rule
\begin{equation}\label{14}
\langle\varepsilon_k, \ s^k\rangle=\left\{\begin{array}{l}0, \quad \varepsilon_k\ne s_k\\
                            1, \quad \varepsilon_k=s_k.
                            \end{array}\right.
\end{equation}
The operation (\ref{14}) is linearly extended to arbitrary chains and cochains.

The coboundary operator $d^c: K^r(4)\rightarrow K^{r+1}(4)$ is defined by
\begin{equation}\label{15}
\langle\partial a, \ \overset{r}{\omega}\rangle=\langle a, \ d^c\overset{r}{\omega}\rangle,
\end{equation}
where $a\in C(4)$ is an $r+1$ dimensional chain. The operator $d^c$ is an analog of the exterior differential.
From the above it follows that
\begin{equation*}\label{}
 d^c\overset{4}{\omega}=0 \quad \mbox{and} \quad d^cd^c\overset{r}{\omega}=0 \quad \mbox{for any} \quad r.
\end{equation*}
Let the difference operator $\Delta_i$ is  given  by
\begin{equation*}\label{}
\Delta_i\omega_k=\omega_{\tau_ik}-\omega_k
\end{equation*}
for any components $\omega_k\in\mathbb{C}$ of $\overset{r}{\omega}\in K^r(4)$. For simplicity of notation we write here $\omega_k$ instead of $\omega_k^{(r)}$ and $\tau_i$ is
  the shift operator  which acts  as
\begin{equation*}\label{}
\tau_ik=(k_0,...\tau
 k_i,...k_3), \quad
  i=0,1,2,3,
  \end{equation*}
  where $\tau$ is defined by (\ref{4}).
Using (\ref{7}) and (\ref{15}) we can calculate
\begin{equation}\label{16}
d^c\overset{0}{\omega}=\sum_k\sum_{i=0}^3(\Delta_i\overset{0}{\omega}_k)e_i^k,
\end{equation}
\begin{equation}\label{17}
d^c\overset{1}{\omega}=\sum_k\sum_{i<j}(\Delta_i\omega_k^j-\Delta_j\omega_k^i)e_{ij}^k,
\end{equation}
\begin{align}\label{18}
d^c\overset{2}{\omega}=\sum_k\big[(\Delta_0\omega_k^{12}-\Delta_1\omega_k^{02}+\Delta_2\omega_k^{01})e_{012}^k\nonumber \\
+(\Delta_0\omega_k^{13}-\Delta_1\omega_k^{03}+\Delta_3\omega_k^{01})e_{013}^k \nonumber \\
+(\Delta_0\omega_k^{23}-\Delta_2\omega_k^{03}+\Delta_3\omega_k^{02})e_{023}^k \nonumber \\
+(\Delta_1\omega_k^{23}-\Delta_2\omega_k^{13}+\Delta_3\omega_k^{12})e_{123}^k\big],
\end{align}
\begin{equation}\label{19}
d^c\overset{3}{\omega}=\sum_k(\Delta_0\omega_k^{123}-\Delta_1\omega_k^{023}+\Delta_2\omega_k^{013}-\Delta_3\omega_k^{012})e^k.
\end{equation}

We now consider   a multiplication of discrete forms which is an analog of the
exterior multiplication for differential forms.  Denote by  $\cup$ this multiplication.
For the basis elements of the one-dimensional complex $K(1)=K$ the $\cup$-multiplication is defined as follows
\begin{equation*}\label{}
x^\kappa\cup x^\kappa=x^\kappa, \quad e^\kappa\cup x^{\tau\kappa}=e^\kappa,
\quad x^\kappa\cup e^\kappa=e^\kappa, \quad \kappa\in{\mathbb Z},
\end{equation*}
supposing the product to be zero in all other case.
To arbitrary basis element of $K(p)$   this definition is  extended by induction on $p$, where ${p=2,3,4}$.
See \cite{Dezin, S3} for details.
 The $\cup$-multiplication can be spread  linearly to  forms.
This definition leads to the following discrete version of the Leibniz rule for differential forms.
\begin{prop}
Let $\varphi$ and $\psi$ be arbitrary forms of $K(4)$.
Then
\begin{equation*}\label{}
 d^c(\varphi\cup\psi)=d^c\varphi\cup\psi+(-1)^r\varphi\cup
d^c\psi,
\end{equation*}
where  $r$ is the degree  of a form $\varphi$.
\end{prop}

Let  $\tilde K(4)$ be a complex of the cochains over the double complex
$\tilde C(4)$, with the coboundary operator $d^c$ defined in it by
(\ref{15}). Hence, $\tilde K(4)$ has the same structure as  $K(4)$.
The definitions of $d^c$  and $\cup$  do not depend on a metric.
 At the same time, to define a discrete analog of the Hodge star operator
$\ast$ we must take into account the  Lorentz metric structure  on $K(4)$.
This means that Definition (\ref{8}) is not suitable for a discrete version of the Hodge star operator.
Define the operation $\ast: K^r(4)\rightarrow \tilde K^{4-r}(4)$ for an arbitrary basis element $s^k_{(r)}=s^{k_0}\otimes  s^{k_1}\otimes s^{k_2}\otimes s^{k_3}$ by the rule
\begin{equation}\label{20}
\ast s^k_{(r)}=Q(k_0)\varepsilon(r)\tilde s^k_{(4-r)},
\end{equation}
where
\begin{equation*}
Q(k_0)=\left\{\begin{array}{l}+1 \quad  \mbox{if} \quad s^{k_0}=x^{k_0} \\
                            -1 \quad  \mbox{if} \quad s^{k_0}= e^{k_0}.
                            \end{array}\right.
\end{equation*}
 This definition makes sense because the formula (\ref{20})  preserves the  Lorentz  signature of metric in our  discrete model. From (\ref{20}) we obtain
\begin{equation}\label{21}
\ast x^k=\ast(x^{k_0}\otimes x^{k_1}\otimes x^{k_2}\otimes x^{k_3})=\tilde e^k,
\end{equation}
\begin{equation}\label{22}
\ast e^k=\ast(e^{k_0}\otimes e^{k_1}\otimes e^{k_2}\otimes e^{k_3})=-\tilde x^k,
\end{equation}
\begin{equation}\label{23}
\ast e_0^k=-\tilde e_{123}^k, \qquad \ast e_1^k=-\tilde e_{023}^k, \qquad
\ast e_2^k=\tilde e_{013}^k, \qquad \ast e_3^k=-\tilde e_{012}^k,
\end{equation}
\begin{align}\label{24}
\ast e_{01}^k&=-\tilde e_{23}^ k, \qquad \ast e_{02}^k=\tilde e_{13}^ k, \qquad \ast e_{03}^k=-\tilde e_{12}^ k, \nonumber \\
\ast e_{12}^k&=\tilde e_{03}^ k, \qquad  \ast e_{13}^k=-\tilde e_{02}^ k, \qquad \ast e_{23}^k=\tilde e_{01}^ k,
\end{align}
\begin{equation}\label{25}
\ast e_{012}^k=-\tilde e_3^ k, \quad \ast e_{013}^k=\tilde e_2^ k, \quad
\ast e_{023}^k=-\tilde e_1^ k, \quad \ast e_{123}^k=-\tilde e_0^k.
\end{equation}
For any $r$-form we extend (\ref{20}) by linearity.

It is easy to check that
 \begin{equation*}
 \ast\ast s^k_{(r)}=(-1)^{r(4-r)+1}s^k_{(r)}=(-1)^{r+1}s^k_{(r)}.
 \end{equation*}
   Consequently,  for each $r$-form $\overset{r}{\omega}\in K^r(4)$ we  have
 \begin{equation}\label{26}
\ast\ast\overset{r}{\omega}=(-1)^{r(4-r)+1}\overset{r}{\omega}=
(-1)^{r+1}\overset{r}{\omega}.
 \end{equation}
 It means that the discrete $\ast$ operation imitates correctly the continual case.
\begin{prop}
Let $a_r\in C(4)$ be an $r$-dimensional chain (\ref{10}).
Then we have
\begin{equation}\label{27}
\langle \tilde a_r, \ \ast\omega\rangle=(-1)^rQ(k_0)\langle \ast^c \tilde a_r, \ \omega\rangle,
\end{equation}
where $\omega\in K^{4-r}(4)$.
\end{prop}
\begin{proof}
The operations $\ast^c$ and $\ast$ are linear. It suffices to prove  (\ref{27}) for  basis elements.
Let $s^k_{(r)}\in K^r(4)$ be an arbitrary $r$-dimensional basis element of $K(4)$ which corresponds to  $s_k^{(r)}\in C(4)$.
By the definition of $\ast^c$ and by (\ref{14}) we have
\begin{equation*}
1=\langle s_k^{(4-r)}, \ s^k_{(4-r)}\rangle=\varepsilon(r)\langle \ast^c \tilde s_k^{(r)}, \ s^k_{(4-r)}\rangle.
 \end{equation*}
 On the other hand:
 \begin{equation*}
1=\langle \tilde s_k^{(r)}, \ \tilde s^k_{(r)}\rangle=Q(k_0)\varepsilon'(r)\langle \tilde s_k^{(r)}, \ \ast s^k_{(4-r)}\rangle
 \end{equation*}
 which yields
 \begin{equation*}
Q(k_0)\varepsilon'(r)\langle \tilde s_k^{(r)}, \ \ast s^k_{(4-r)}\rangle=\varepsilon(r)\langle \ast^c \tilde s_k^{(r)}, \ s^k_{(4-r)}\rangle.
 \end{equation*}
 Multiplying by $Q(k_0)\varepsilon'(r)$ both sides of the above we find
 \begin{equation*}
\langle \tilde s_k^{(r)}, \ \ast s^k_{(4-r)}\rangle=Q(k_0)\varepsilon'(r)\varepsilon(r)\langle \ast^c \tilde s_k^{(r)}, \ s^k_{(4-r)}\rangle,
 \end{equation*}
 since $Q^2(k_0)=(\varepsilon'(r))^2=1$.
 Thus, using (\ref{9}) we obtain (\ref{27}).
 \end{proof}

Let $V\subset C(4)$ be a finite set of $4$-dimensional basis elements of the complex $C(4)$.
Define $V$ as follows
\begin{equation}\label{28}
V=\sum_ke_k, \qquad k=(k_0,k_1,k_2,k_3),\qquad k_i=1,2, ...,N_i,
\end{equation}
where $N_i\in \mathbb{N}$ is a fixed number and $i=0,1,2,3$. We agree that
in what follows the subscripts $k_i$ always run the set
of values indicated in (\ref{28}).

We set
\begin{equation*}\label{}
 V_r=\sum_k\sum_{(r)}s_k^{(r)}\otimes\ast^c s_k^{(r)}.
\end{equation*}
 For example,
\begin{equation*}
 V_1=\sum_k\sum_{i=0}^3 e_k^i\otimes\ast^c e_k^i=\sum_k(e_k^0\otimes\tilde e_k^{123}-
 e_k^1\otimes\tilde e_k^{023}+e_k^2\otimes\tilde e_k^{013}-e_k^3\otimes\tilde
 e_k^{012})
 \end{equation*}
 and
 \begin{equation*}
 V_2=\sum_k\sum_{i<j}^3 e_k^{ij}\otimes\ast^c e_k^{ij}=\sum_k(e_k^{01}\otimes\tilde e_k^{23}-e_k^{02}\otimes\tilde e_k^{13}+
 e_k^{03}\otimes\tilde e_k^{12}+e_k^{12}\otimes\tilde e_k^{03}-e_k^{13}\otimes\tilde e_k^{02}+e_k^{23}\otimes\tilde e_k^{01}),
 \end{equation*}
 where $e_k^i$ and $e_k^{ij}$ are given by (\ref{5}) and (\ref{6}).

 Let
 \begin{equation*}
 \mathbb{V}=\sum_{r=0}^4 V_r.
 \end{equation*}
 For any $r$-forms $\varphi, \omega\in K^r(4)$ we define the inner
 product  $(\ , \ )_V$ by
 \begin{align}\label{29}\notag
 (\varphi ,\ \omega)_V&=\langle\mathbb{V}, \ \varphi\otimes\ast\overline{\omega}\rangle=\langle V_r, \ \varphi\otimes\ast\overline{\omega}\rangle\\
 &=\sum_k\sum_{(r)}
 \langle s_k^{(r)},\
\varphi \rangle\langle\ast^c s_k^{(r)}, \ \ast\overline{\omega}\rangle,
\end{align}
where $\overline{\omega}$ denotes the complex conjugate of the form $\omega$, i. e., $\overline{\omega}= \sum_k \sum_{(r)}\overline{\omega}_k^{(r)}s_{(r)}^k$.
For the forms of different degrees the product (\ref{29}) is set equal to zero.
For example, if $\overset{1}{\varphi}, \ \overset{1}{\omega}\in K^1(4)$ then we obtain
\begin{equation*}\label{}
(\overset{1}{\varphi}, \ \overset{1}{\omega})_V=\sum_k\big[-\varphi_k^0\overline{\omega}_k^0+\varphi_k^1\overline{\omega}_k^1+\varphi_k^2\overline{\omega}_k^2+
\varphi_k^3\overline{\omega}_k^3\big].
\end{equation*}
It should be noted that in the definition of the inner product a role of $\cup$-multiplication is now played by the tensor multiplication (cf. \cite{S3}).   In (\ref{29}) the  Lorentz  matric structure is still captured. Using (\ref{14}) and (\ref{21})--(\ref{25})
 we obtain
\begin{align*}\label{}
(\overset{0}{\omega}, \ \overset{0}{\omega})_V&=\sum_k|\overset{0}{\omega}_k|^2, \\
(\overset{1}{\omega}, \ \overset{1}{\omega})_V&=\sum_k\big(-|\omega_k^0|^2+|\omega_k^1|^2+|\omega_k^2|^2+|\omega_k^3|^2\big), \\
(\overset{2}{\omega}, \ \overset{2}{\omega})_V&=\sum_k\big(-|\omega_k^{01}|^2-|\omega_k^{02}|^2-|\omega_k^{03}|^2+|\omega_k^{12}|^2+|\omega_k^{13}|^2+
|\omega_k^{23}|^2\big), \\
(\overset{3}{\omega}, \ \overset{3}{\omega})_V&=\sum_k\big(-|\omega_k^{012}|^2-|\omega_k^{013}|^2-|\omega_k^{023}|^2+|\omega_k^{123}|^2\big), \\
(\overset{4}{\omega}, \ \overset{4}{\omega})_V&=-\sum_k|\overset{4}{\omega}_k|^2.
\end{align*}
The inner product makes it possible to define the adjoint of
$d^c$, denoted $\delta^c$.
\begin{prop} For any $(r-1)$-form $\varphi$ and $r$-form $\omega$ we have
\begin{equation}\label{30}
(d^c\varphi, \ \omega)_V=\langle\partial\mathbb{V}, \ \varphi\otimes\ast\overline{\omega}\rangle+( \varphi, \ \delta^c\omega)_V,
\end{equation} where
 \begin{equation}\label{31}
 \delta^c=(-1)^{r}\ast^{-1}d^c\ast
\end{equation}
and $\ast\ast^{-1}=I$.
\end{prop}
\begin{proof} The proof is a computation. From the definition
(\ref{15}) it follows that  (\ref{7})  induces the similar
relation for the coboundary operator $d^c$ on forms:
\begin{equation*}
d^c(\varphi\otimes\ast\omega)=d^c\varphi\otimes\ast\omega+(-1)^{r-1}\varphi\otimes
d^c(\ast\omega).
\end{equation*}
Using this we compute
\begin{align*}
(d^c\varphi, \ \omega)_V&=\langle\mathbb{V}, \
d^c\varphi\otimes\ast\overline{\omega}\rangle=\langle V_r, \
d^c\varphi\otimes\ast\overline{\omega}\rangle\\&= \langle\mathbb{V}, \
d^c(\varphi\otimes\ast\overline{\omega})\rangle-(-1)^{r-1}\langle\mathbb{V}, \
\varphi\otimes d^c(\ast\overline{\omega})\rangle\\&= \langle\partial \mathbb{V}, \
\varphi\otimes\ast\overline{\omega}\rangle+(-1)^{r}\langle\mathbb{V}, \
\varphi\otimes\ast(\ast^{-1}d^c\ast\overline{\omega})\rangle\\&= \langle\partial \mathbb{V}, \
\varphi\otimes\ast\overline{\omega}\rangle+\langle V_{r-1}, \
\varphi\otimes\ast(\delta^c\overline{\omega})\rangle.
 \end{align*}
 It immediately follows (\ref{30}).
\end{proof}
Relation (\ref{30}) is a discrete analog of the Green formula.
From (\ref{26}) we infer
\begin{equation*}
\ast^{-1}=(-1)^{r(4-r)+1}\ast=(-1)^{r+1}\ast.
\end{equation*}
Putting this in (\ref{31}) we obtain
 \begin{equation}\label{32}
 \delta^c=\ast d^c\ast.
\end{equation}
This makes it clear that the operator $\delta^c: K^{r+1}(4) \rightarrow K^r(4)$ is a discrete analog of the codifferential $\delta$.
  For the 0-form (\ref{12}) we have $\delta^c\overset{0}{\omega}=0$.
 Note that the difference expression for   $\delta^c$ is slightly different than that given in \cite{S3}.
 Using (\ref{16})--(\ref{19}) and (\ref{21})--(\ref{25}) we can calculate
\begin{equation}\label{33}
\delta^c\overset{1}{\omega}=\sum_k(\Delta_0\omega_k^{0}-\Delta_1\omega_k^{1}-\Delta_2\omega_k^{2}-\Delta_3\omega_k^{3})x^k,
\end{equation}
\begin{align}\label{34} \nonumber
\delta^c\overset{2}{\omega}=\sum_k\big[(\Delta_1\omega_k^{01}+\Delta_2\omega_k^{02}+\Delta_3\omega_k^{03})e_{0}^k\\ \nonumber
+(\Delta_0\omega_k^{01}+\Delta_2\omega_k^{12}+\Delta_3\omega_k^{13})e_{1}^k\\ \nonumber
+(\Delta_0\omega_k^{02}-\Delta_1\omega_k^{12}+\Delta_3\omega_k^{23})e_{2}^k\\
+(\Delta_0\omega_k^{03}-\Delta_1\omega_k^{13}-\Delta_2\omega_k^{23})e_{3}^k\big],
\end{align}
\begin{align}\label{35} \nonumber
\delta^c\overset{3}{\omega}=\sum_k\big[(-\Delta_2\omega_k^{012}-\Delta_3\omega_k^{013})e_{01}^k+
(\Delta_1\omega_k^{012}-\Delta_3\omega_k^{023})e_{02}^k\\ \nonumber
+(\Delta_1\omega_k^{013}+\Delta_2\omega_k^{023})e_{03}^k
+(\Delta_0\omega_k^{012}-\Delta_3\omega_k^{123})e_{12}^k\\
+(\Delta_0\omega_k^{013}+\Delta_2\omega_k^{123})e_{13}^k
+(\Delta_0\omega_k^{023}-\Delta_1\omega_k^{123})e_{23}^k\big],
\end{align}
\begin{align}\label{36}
\delta^c\overset{4}{\omega}=\sum_k\big[(\Delta_3\overset{4}{\omega}_k)e_{012}^k-(\Delta_2\overset{4}{\omega}_k)e_{013}^k\\ \nonumber
+(\Delta_1\overset{4}{\omega}_k)e_{023}^k+(\Delta_0\overset{4}{\omega}_k)e_{123}^k\big].
\end{align}
It is obvious that
$\delta^c\delta^c\overset{r}{\omega}=0$ \ for any $r=1,2,3,4.$
The linear map
\begin{equation*}
\Delta^c=-(d^c\delta^c+\delta^cd^c): \ K^r(4) \rightarrow K^r(4)
\end{equation*}
is called a discrete analogue of the Laplacian (\ref{2}). It is clear that
\begin{equation}\label{37}
-(d^c\delta^c+\delta^cd^c)=(d^c-\delta^c)^2=\mathrm{i}(d^c+\delta^c)^2.
\end{equation}

Finally, let us introduce the following  operation
\begin{equation*}
\tilde\iota: K^r(4)
\rightarrow \tilde K^r(4), \qquad \tilde\iota: \tilde K^r(4)
\rightarrow K^r(4)
\end{equation*}
by setting
\begin{equation}\label{38}
 \tilde\iota s_{(r)}^k= \tilde s_{(r)}^k, \qquad \tilde\iota\tilde s_{(r)}^k=  s_{(r)}^k,
\end{equation}
where $s_{(r)}^k$ and $\tilde s_{(r)}^k$ are  basis elements of
$K^r(4)$ and $\tilde K^r(4)$ respectively.  Consequently, for an $r$-form  $\varphi\in K^r(4)$
we have \ $\tilde\iota\varphi=\tilde\varphi$. \ Recall that the components
 of $\tilde\varphi\in \tilde K^r(4)$ and  $\varphi\in
K^r(4)$ are the same.
\begin{prop} The following hold
\begin{equation}\label{39}
\tilde\iota^2=I, \qquad \tilde\iota\ast=\ast\tilde\iota, \qquad
\tilde\iota d^c=d^c\tilde\iota, \qquad \tilde\iota \delta^c=\delta^c\tilde\iota.
\end{equation}
\end{prop}
\begin{proof}
The proof immediately follows from definitions of the
corresponding operations.
\end{proof}

 \section{Discrete Dirac-K\"{a}hler equation and chiral symmetry}
Let us introduce a discrete inhomogeneous form as follows
\begin{equation}\label{40}
\Omega=\sum_{r=0}^4\overset{r}{\omega},
\end{equation}
where $\overset{r}{\omega}$ is given by (\ref{12}) and (\ref{13}).
Due to (\ref{37}) a discrete analog of the Dirac-K\"{a}hler equation (\ref{3}) can be defined as
 \begin{equation}\label{41}
\mathrm{i}(d^c+\delta^c)\Omega=m\Omega.
\end{equation}
We can write this equation more explicitly by separating its homogeneous components as
\begin{align}\label{42} \nonumber
\mathrm{i}\delta^c\overset{1}{\omega}=m\overset{0}{\omega},\\ \nonumber
\mathrm{i}(d^c\overset{0}{\omega}+\delta^c\overset{2}{\omega})=m\overset{1}{\omega},\\
\mathrm{i}(d^c\overset{1}{\omega}+\delta^c\overset{3}{\omega})=m\overset{2}{\omega},\\ \nonumber
\mathrm{i}(d^c\overset{2}{\omega}+\delta^c\overset{4}{\omega})=m\overset{3}{\omega},\\ \nonumber
\mathrm{i}d^c\overset{3}{\omega}=m\overset{4}{\omega}.
\end{align}
This set of equations  can be expressed in terms of difference equations.
Substituting (\ref{16})--(\ref{19}) and (\ref{33})--(\ref{36}) into (\ref{42}) we obtain
\begin{align*}\label{}
\mathrm{i}(\Delta_0\omega_k^{0}-\Delta_1\omega_k^{1}-\Delta_2\omega_k^{2}-\Delta_3\omega_k^{3})=m\overset{0}{\omega}_k,\\
\mathrm{i}(\Delta_0\overset{0}{\omega}_k+\Delta_1\omega_k^{01}+\Delta_2\omega_k^{02}+\Delta_3\omega_k^{03})=m\omega_k^0,\\
\mathrm{i}(\Delta_1\overset{0}{\omega}_k+\Delta_0\omega_k^{01}+\Delta_2\omega_k^{12}+\Delta_3\omega_k^{13})=m\omega_k^1,\\
\mathrm{i}(\Delta_2\overset{0}{\omega}_k+\Delta_0\omega_k^{02}-\Delta_1\omega_k^{12}+\Delta_3\omega_k^{23})=m\omega_k^2,\\
\mathrm{i}(\Delta_3\overset{0}{\omega}_k+\Delta_0\omega_k^{03}-\Delta_1\omega_k^{13}-\Delta_2\omega_k^{23})=m\omega_k^3,\\
\mathrm{i}(\Delta_0\omega_k^1-\Delta_1\omega_k^0-\Delta_2\omega_k^{012}-\Delta_3\omega_k^{013})=m\omega_k^{01},\\
\mathrm{i}(\Delta_0\omega_k^2-\Delta_2\omega_k^0+\Delta_1\omega_k^{012}-\Delta_3\omega_k^{023})=m\omega_k^{02},\\
\mathrm{i}(\Delta_0\omega_k^3-\Delta_3\omega_k^0+\Delta_1\omega_k^{013}+\Delta_2\omega_k^{023})=m\omega_k^{03},\\
\mathrm{i}(\Delta_1\omega_k^2-\Delta_2\omega_k^1+\Delta_0\omega_k^{012}-\Delta_3\omega_k^{123})=m\omega_k^{12},\\
\mathrm{i}(\Delta_1\omega_k^3-\Delta_3\omega_k^1+\Delta_0\omega_k^{013}+\Delta_2\omega_k^{123})=m\omega_k^{13},\\
\mathrm{i}(\Delta_2\omega_k^3-\Delta_3\omega_k^2+\Delta_0\omega_k^{023}-\Delta_1\omega_k^{123})=m\omega_k^{23}, \\
\mathrm{i}(\Delta_0\omega_k^{12}-\Delta_1\omega_k^{02}+\Delta_2\omega_k^{01}+\Delta_3\overset{4}{\omega}_k)=m\omega_k^{012},\\
\mathrm{i}(\Delta_0\omega_k^{13}-\Delta_1\omega_k^{03}+\Delta_3\omega_k^{01}-\Delta_2\overset{4}{\omega}_k)=m\omega_k^{013},\\
\mathrm{i}(\Delta_0\omega_k^{23}-\Delta_2\omega_k^{03}+\Delta_3\omega_k^{02}+\Delta_1\overset{4}{\omega}_k)=m\omega_k^{023},\\
\mathrm{i}(\Delta_1\omega_k^{23}-\Delta_2\omega_k^{13}+\Delta_3\omega_k^{12}+\Delta_0\overset{4}{\omega}_k)=m\omega_k^{123},\\
\mathrm{i}(\Delta_0\omega_k^{123}-\Delta_1\omega_k^{023}+\Delta_2\omega_k^{013}-\Delta_3\omega_k^{012})=m\overset{4}{\omega}_k.
\end{align*}

Let us introduce the modified star operator $\star$ on  inhomogeneous forms (\ref{40}) by the rule
\begin{equation}\label{43}
\star\Omega=\mathrm{i}\ast B\Omega,
\end{equation}
where $B$ is the  main antiautomorphism (see \cite{Reuter} for more details) which acts on $\Omega$ according to
\begin{equation}\label{44}
B\Omega=\sum_{r=0}^4(-1)^{\frac{r(r-1)}{2}}\overset{r}{\omega}.
\end{equation}
\begin{prop}
The operator $\star$ is an involution, i.e.,
 \begin{equation}\label{45}
\star\star\Omega=\Omega.
\end{equation}
 \end{prop}
 \begin{proof}
 Acting twice with $\ast B$ on an $r$-form and using (\ref{26}) we obtain
 \begin{equation*}
  \ast B\ast B\overset{r}{\omega}=(-1)^{\frac{r(r-1)}{2}}\ast B(\ast\overset{r}{\omega})=(-1)^{\frac{r(r-1)}{2}}(-1)^{\frac{(4-r)(4-r-1)}{2}}\ast \ast\overset{r}{\omega}=-\overset{r}{\omega}.
 \end{equation*}
 \end{proof}

\begin{prop}
 Let   $\Omega$ be an inhomogeneous form. Then we have
 \begin{equation}\label{46}
\star(d^c+\delta^c)\Omega+(d^c+\delta^c)\star\Omega=0.
\end{equation}
 \end{prop}
 \begin{proof}
 By  (\ref{44}) for a homogeneous component of $\Omega$ we have
  \begin{equation}\label{47}
  \star\overset{r}{\omega}=(-1)^{\frac{r(r-1)}{2}}\mathrm{i}\ast\overset{r}{\omega}.
  \end{equation}
  Let $r=1$.  Then
  \begin{align*}
\star(d^c+\delta^c)\overset{1}{\omega}&=-\mathrm{i}\ast d^c\overset{1}{\omega}+\mathrm{i}\ast\delta^c\overset{1}{\omega}=-\mathrm{i}\ast d^c\ast\ast\overset{1}{\omega}+
\mathrm{i}\ast\ast d^c\ast\overset{1}{\omega}\\&=
-\delta^c(\mathrm{i}\ast\overset{1}{\omega})-d^c(\mathrm{i}\ast\overset{1}{\omega})=-\delta^c\star\overset{1}{\omega}-d^c\star\overset{1}{\omega}=
-(\delta^c+d^c)\star\overset{1}{\omega}.
\end{align*}
In the same way we obtain
\begin{align*}
\star(d^c+\delta^c)\overset{2}{\omega}&=-\mathrm{i}\ast d^c\overset{2}{\omega}+\mathrm{i}\ast\delta^c\overset{2}{\omega}=\mathrm{i}\ast d^c\ast\ast\overset{2}{\omega}+
\mathrm{i}\ast\ast d^c\ast\overset{2}{\omega}\\&=
\delta^c(\mathrm{i}\ast\overset{2}{\omega})+d^c(\mathrm{i}\ast\overset{2}{\omega})=-\delta^c\star\overset{2}{\omega}-d^c\star\overset{2}{\omega}=
-(\delta^c+d^c)\star\overset{2}{\omega}
\end{align*}
and
\begin{equation*}
\star(d^c+\delta^c)\overset{3}{\omega}=-(\delta^c+d^c)\star\overset{3}{\omega}.
\end{equation*}
Next
\begin{equation*}
\star(d^c+\delta^c)\overset{0}{\omega}=\star d^c\overset{0}{\omega}=-\mathrm{i}\ast d^c\ast\ast\overset{0}{\omega}=-\delta^c(\mathrm{i}\ast\overset{0}{\omega})=
-\delta^c\star\overset{0}{\omega}
\end{equation*}
and
\begin{equation*}
\star(d^c+\delta^c)\overset{4}{\omega}=\star \delta^c\overset{4}{\omega}=-\mathrm{i}\ast\ast d^c\ast\overset{4}{\omega}=-d^c(\mathrm{i}\ast\overset{4}{\omega})=
-d^c\star\overset{4}{\omega}.
\end{equation*}
Hence
\begin{equation*}
\star(d^c+\delta^c)\Omega=\sum_{r=0}^4\star(d^c+\delta^c)\overset{r}{\omega}=-\sum_{r=0}^4(d^c+\delta^c)\star\overset{r}{\omega}=-(d^c+\delta^c)\star\Omega.
\end{equation*}
 \end{proof}
 By virtue of (\ref{45}) and (\ref{46}) we can claim that
 the operator $\star$ plays the same role as the fifth gamma matrix $\gamma^5$ in continual Dirac theory. Therefore, we can consider the chirality of our discrete model
 with respect to this operator.

We say that an inhomogeneous form $\Omega$ is self-dual or anti-self-dual if
\begin{equation}\label{48}
\tilde{\iota}\star\Omega=\Omega \qquad  \mbox{or} \qquad \tilde{\iota}\star\Omega=-\Omega,
\end{equation}
where $\tilde{\iota}$ is defined by (\ref{38}).
The first equation of (\ref{48}) is equivalent to the following equations
\begin{equation*}\label{}
\tilde{\iota}\star\overset{r}{\omega}=\overset{4-r}{\omega},
\end{equation*}
where  $r=0,1,2,3,4$.
Hence, from  (\ref{47}) for a self-dual form we obtain
\begin{equation}\label{49}
\tilde{\iota}\ast\overset{0}{\omega}=-\mathrm{i}\overset{4}{\omega}, \quad  \tilde{\iota}\ast\overset{1}{\omega}=-\mathrm{i}\overset{3}{\omega}, \quad
\tilde{\iota}\ast\overset{2}{\omega}=\mathrm{i}\overset{2}{\omega}, \quad \tilde{\iota}\ast\overset{3}{\omega}=\mathrm{i}\overset{1}{\omega}, \quad
 \tilde{\iota}\ast\overset{4}{\omega}=-\mathrm{i}\overset{0}{\omega}.
\end{equation}
The same equations can be drawn for an anti-self-dual form.

\begin{prop}
 If  the self-dual or anti-self-dual form $\Omega$ is a solution of the discrete Dirac-K\"{a}hler equation, then $\Omega$ is trivial, i. e., $\Omega=0$.
 \end{prop}
\begin{proof}
Let $\Omega$ be a self-dual form. By (\ref{49}) from the first and fifth equations   of (\ref{42})  we have
\begin{align*}
m\overset{0}{\omega}=\mathrm{i}\delta^c\overset{1}{\omega}=\mathrm{i}\ast d^c\ast\overset{1}{\omega}=\mathrm{i}\ast d^c(-\tilde{\iota}\mathrm{i}\overset{3}{\omega})=\tilde{\iota}\ast d^c\overset{3}{\omega}\\
=\tilde{\iota}\ast(-\mathrm{i}m\overset{4}{\omega})=-\mathrm{i}m(\tilde{\iota}\ast\overset{4}{\omega})=-m\overset{0}{\omega}.
\end{align*}
Here we used (\ref{39}).
Since $m>0$ we obtain $\overset{0}{\omega}=0$. It follows immediately that $\overset{4}{\omega}=0$.
Using this,  from the second  and fourth equations   of  (\ref{42}) we obtain again
\begin{align*}
m\overset{1}{\omega}=\mathrm{i}(d^c\overset{0}{\omega}+\delta^c\overset{2}{\omega})=\mathrm{i}\ast d^c\ast\overset{2}{\omega}=
\mathrm{i}\ast d^c(\tilde{\iota}\mathrm{i}\overset{2}{\omega})=-\tilde{\iota}\ast d^c\overset{2}{\omega}\\
=-\tilde{\iota}\ast(-\mathrm{i}m\overset{3}{\omega}-\delta^c\overset{4}{\omega})=\mathrm{i}m(\tilde{\iota}\ast\overset{3}{\omega})= \mathrm{i}^2m\overset{1}{\omega}=-m\overset{1}{\omega}.
\end{align*}
Thus, $\overset{1}{\omega}=0$ and this yields  $\overset{3}{\omega}=0$. Furthermore, according to the third equation of  (\ref{42})  we have  $\overset{2}{\omega}=0$.
  \end{proof}

  An inhomogeneous form $\Omega$ decomposes into its self-dual and anti-self-dual parts with respect to the action of $\star$ as follows
  \begin{equation*}\label{}
\Omega=\Omega^{+}+\Omega^{-},
\end{equation*}
where
\begin{equation*}\label{}
\Omega^\pm=\frac{1}{2}(\Omega\pm\tilde{\iota}\star\Omega).
\end{equation*}
It is clear that $\Omega^+$ is self-dual and $\Omega^-$ is anti-self-dual.
 The self-dual and anti-self-dual components of $\Omega$ correspond to the chiral right and chiral left parts of a Dirac fermion.
 Denote by $\overset{r}{\omega}^{\pm}$  a homogeneous component of $\Omega^{\pm}$.
Using (\ref{47}) we can write these components more explicitly  as
\begin{align}\notag\label{50}
\overset{0}{\omega}^{\pm}=\frac{1}{2}(\overset{0}{\omega}\pm \mathrm{i}\tilde{\iota}\ast\overset{4}{\omega}), \qquad \overset{1}{\omega}^{\pm}=\frac{1}{2}(\overset{1}{\omega}\mp \mathrm{i}\tilde{\iota}\ast\overset{3}{\omega}), \qquad \overset{2}{\omega}^{\pm}=\frac{1}{2}(\overset{2}{\omega}\mp \mathrm{i}\tilde{\iota}\ast\overset{2}{\omega}),\\
\overset{3}{\omega}^{\pm}=\frac{1}{2}(\overset{3}{\omega}\pm \mathrm{i}\tilde{\iota}\ast\overset{1}{\omega}), \qquad
\overset{4}{\omega}^{\pm}=\frac{1}{2}(\overset{4}{\omega}\pm \mathrm{i}\tilde{\iota}\ast\overset{0}{\omega}).
\end{align}
\begin{prop}
 If   $\Omega$ is a solution of the massless discrete Dirac-K\"{a}hler equation
 \begin{equation}\label{51}
(d^c+\delta^c)\Omega=0,
\end{equation}
  then so do both  $\Omega^+$ and $\Omega^-$.
 \end{prop}
\begin{proof}
It suffices to prove the claim for one of Eqs. (\ref{42}), say for the second one.  Let
\begin{equation*}\label{}
\mathrm{i}(d^c\overset{0}{\omega}+\delta^c\overset{2}{\omega})=0.
\end{equation*}
Then, by (\ref{50}),   for the corresponding homogeneous components of  $\Omega^+$ we obtain
\begin{align*}
\mathrm{i}(d^c\overset{0}{\omega}^{+}+\delta^c\overset{2}{\omega}^{+})&=\frac{1}{2}\mathrm{i}(d^c(\overset{0}{\omega}+\mathrm{i}\tilde{\iota}\ast\overset{4}{\omega})+
\delta^c(\overset{2}{\omega}-\mathrm{i}\tilde{\iota}\ast\overset{2}{\omega}))\\
&=\frac{1}{2}\mathrm{i}(d^c\overset{0}{\omega}+\delta^c\overset{2}{\omega})+\frac{1}{2}\tilde{\iota}(-d^c(\ast\overset{4}{\omega})+\delta^c(\ast\overset{2}{\omega}))\\
&=\frac{1}{2}\tilde{\iota}(-\ast\ast d^c\ast\overset{4}{\omega}+\ast d^c\ast\ast\overset{2}{\omega})=
-\frac{1}{2}\tilde{\iota}\ast(\delta^c\overset{4}{\omega}+d^c\overset{2}{\omega})=0.
\end{align*}
Here we used  (\ref{26}), (\ref{32}) and (\ref{39}).  The other cases are similar.
\end{proof}
Proposition 3.4 gives rise to the chiral symmetry of our discrete model. This means that equation  (\ref{51}) is invariant under the transformation
\begin{equation*}\label{}
\Omega\longrightarrow\Omega\pm\tilde{\iota}\star\Omega.
\end{equation*}
This transformation is equivalent to
\begin{equation}\label{52}
\overset{r}{\omega}\longrightarrow\overset{r}{\omega}\pm\tilde{\iota}\star\overset{4-r}{\omega}
\end{equation}
and the equation
\begin{equation*}\label{}
d^c\overset{r-1}{\omega}+\delta^c\overset{r+1}{\omega}=0
\end{equation*}
is invariant under the transformation (\ref{52}) for any $r=0,1,2,3,4$.
\begin{prop}
 If   $\Omega$ is a solution of the  discrete Dirac-K\"{a}hler equation (\ref{41}) then we have
 \begin{equation}\label{53}
\mathrm{i}(d^c+\delta^c)\Omega^+=m\Omega^-,  \qquad \mathrm{i}(d^c+\delta^c)\Omega^-=m\Omega^+.
\end{equation}
 \end{prop}
\begin{proof}
Substituting (\ref{50}) into the first and second equations of (\ref{42}) we obtain
\begin{align*}
\mathrm{i}\delta^c\overset{1}{\omega}^{+}&=\frac{1}{2}\mathrm{i}(\delta^c\overset{1}{\omega}- \mathrm{i}\tilde{\iota}\delta^c\ast\overset{3}{\omega})=
\frac{1}{2}(m\overset{0}{\omega}+\tilde{\iota}\ast d^c\ast\ast\overset{3}{\omega})=
\frac{1}{2}(m\overset{0}{\omega}+\tilde{\iota}\ast d^c\overset{3}{\omega})\\&=
\frac{1}{2}(m\overset{0}{\omega}+\tilde{\iota}\ast(-\mathrm{i}m\overset{4}{\omega}))=\frac{1}{2}m(\overset{0}{\omega}-\mathrm{i}\tilde{\iota}\ast\overset{4}{\omega})=
\frac{1}{2}m(\overset{0}{\omega}-\tilde{\iota}\star\overset{4}{\omega})=m\overset{0}{\omega}^{-}
\end{align*}
and
\begin{align*}
\mathrm{i}(d^c\overset{0}{\omega}^{+} +\delta^c\overset{2}{\omega}^{+})&=\frac{1}{2}\mathrm{i}d^c(\overset{0}{\omega}+\mathrm{i}\tilde{\iota}\ast\overset{4}{\omega})+\frac{1}{2}\mathrm{i}\delta^c(\overset{2}{\omega}- i\tilde{\iota}\ast\overset{2}{\omega})\\&= \frac{1}{2}(m\overset{1}{\omega}-
\tilde{\iota}d^c\ast\overset{4}{\omega}+\tilde{\iota}\delta^c\ast\overset{2}{\omega})=\frac{1}{2}(m\overset{1}{\omega}-
\tilde{\iota}\ast\ast d^c\ast\overset{4}{\omega}+\tilde{\iota}\ast d^c\ast\ast\overset{2}{\omega})\\&=
\frac{1}{2}(m\overset{1}{\omega}-
\tilde{\iota}\ast\delta^c\overset{4}{\omega}-\tilde{\iota}\ast d^c\overset{2}{\omega})=
\frac{1}{2}(m\overset{1}{\omega}-
\tilde{\iota}\ast(\delta^c\overset{4}{\omega}+d^c\overset{2}{\omega}))\\&=
\frac{1}{2}(m\overset{1}{\omega}-
\tilde{\iota}\ast(-\mathrm{i}m\overset{3}{\omega}))=
\frac{1}{2}m(\overset{1}{\omega}+\mathrm{i}\tilde{\iota}\ast\overset{3}{\omega})=
\frac{1}{2}m(\overset{1}{\omega}-\tilde{\iota}\star\overset{3}{\omega})=m\overset{1}{\omega}^{-}.
\end{align*}
Here we used also (\ref{26}), (\ref{32}) and (\ref{47}).
Similar calculations give
\begin{align*}
\mathrm{i}(d^c\overset{1}{\omega}^{+}+\delta^c\overset{3}{\omega}^{+})=m\overset{2}{\omega}^{-}, \quad
\mathrm{i}(d^c\overset{2}{\omega}^{+}+\delta^c\overset{4}{\omega}^{+})=m\overset{3}{\omega}^{-}, \quad
\mathrm{i}d^c\overset{3}{\omega}^{+}=m\overset{4}{\omega}^{-}.
\end{align*}
Consequently, the first equation of (\ref{53}) is true.
 The remaining case is similar.
\end{proof}
Proposition 3.5 claims that in the massive case the operator $\mathrm{i}(d^c+\delta^c)$ flips the chirality.
It is known that in the continual Dirac theory the left-hand fermions turn into right-hand fermions after acting the Dirac operator and vice versa.
Thus we have the same result in the discrete case.

\end{document}